\documentclass[onecolumn,10pt]{IEEEtran}

\usepackage{amsmath,amsfonts,amsthm, amssymb, enumerate}
\usepackage{cite}
\usepackage{lipsum}
\usepackage[justification=centering]{caption}
\usepackage{amssymb,graphicx,subfigure}
\usepackage[bottom]{footmisc}
\usepackage[dvipsnames]{xcolor}
\usepackage{balance}

\usepackage{array,multirow}
\newlength{\Oldarrayrulewidth}

\newtheorem{thm}{Theorem}[section]
\newtheorem{lem}[thm]{Lemma}

\usepackage{algorithm,algpseudocode}
\algnewcommand\Input{\item[{\textbf{Input:}}]}
\algnewcommand\Output{\item[{\textbf{Output:}}]}
\algnewcommand\Initialize{\item[{\textbf{Initialize:}}]}
\algnewcommand{\return}[1]{
  \State \textbf{return:}
  \Statex \hspace*{\algorithmicindent}\parbox[t]{.8\linewidth}{\raggedright #1}
}

\begin{document}
\title{Bayesian Approach with Extended Support Estimation  for Sparse Regression}

\author{Kyung-Su~Kim,~Sae-Young Chung\\
School of Electrical Engineering, Korea Advanced Institute of Science and Technology\\
Email: {kyungsukim, schung}@kaist.ac.kr}

\markboth{ }
{Shell \MakeLowercase{\textit{et al.}}: Bare Demo of IEEEtran.cls for Journals}

\maketitle

\begin{abstract}
A greedy algorithm called Bayesian multiple matching pursuit (BMMP) is proposed to estimate a sparse signal vector and its support given $m$ linear measurements. Unlike the maximum a posteriori (MAP) support detection, which was proposed by Lee to estimate the support by selecting an index with the maximum likelihood ratio of the correlation given by a normalized version of the orthogonal matching pursuit (OMP), the proposed method uses the correlation given by the matching pursuit proposed by Davies and Eldar. BMMP exploits the diversity gain to estimate the support by considering multiple support candidates, each of which is obtained by iteratively selecting an index set with a size different for each candidate. In particular, BMMP considers an extended support estimate whose maximal size is $m$ in the process to obtain each of the support candidates. It is observed that BMMP outperforms other state-of-the-art methods and approaches the ideal limit of the signal sparsity in our simulation setting.
\end{abstract}

\IEEEpeerreviewmaketitle

\section{Introduction}
\label{sec_intro}
Sparse linear regression, refered to as compressed sensing (CS) \cite{candes2006compressive,donoho2006compressed}, has been widely studied in many fields ranging from denoising \cite{metzler2016denoising} to super-resolution restoration \cite{yang2010image,heckel2016super}. {A signal vector $x^* \in \mathbb{R}^{n}$ is called $k$-sparse if it has at most $k$ non-zero elements.} Then, the objective of sparse linear regression is to recover a $k$-sparse signal vector $x^* \in \mathbb{R}^{n}$ from the given measurement vector $y \in \mathbb{R}^{m}$, such that 
\begin{align}\nonumber
y =\Phi x^* + w,
\end{align}
 where $\Phi \in \mathbb{R}^{m \times n}$ is a known sensing matrix with $m<n$ and a noise vector $w \in \mathbb{R}^{m }$. 

To improve recovery performance of the sparse signal $x^*$ and its support $\Omega$, the set of non-zero indices in $x^*$, a method called MAP support detection \cite{lee2016map} to recover a partial support of $x^*$ was proposed by exploiting the maximum a posteriori (MAP) estimation. This method determines whether the index $i$ belongs to $\Omega$ by using a likelihood ratio test under the true and null hypotheses given a correlation between the residual vector and its $l_2$-normalized column vector $\phi_{i}/\left\|\phi_{i}\right\|$ in $\Phi:=[\phi_1,...,\phi_n] $. By applying this method to existing algorithms such as generalized OMP (gOMP) \cite{wang2012generalized}, compressive sampling matching pursuit (CoSaMP) \cite{needell2009cosamp}, and subspace pursuit (SP) \cite{dai2009subspace}, some greedy methods were proposed and shown to outperform other existing methods\cite{lee2016map}. 

We propose a new MAP-based method to estimate partial support by modifying the MAP support detection to use the correlation term given in rank aware order recursive matching pursuit (RA-ORMP) \cite{davies2012rank}. This proposed scheme improves the existing MAP support detection by additionally utilizing an orthogonal complement onto each column vector in $\Phi$ to calculate the correlation, since the correlation can be better represented through this consideration. 

Based on this method to estimate a partial support, we develop an algorithm called Bayesian multiple matching pursuit (BMMP) to recover $x^*$ and its support $\Omega$.
BMMP also uses the following four techniques to enhance the performance further. Simulation results show that BMMP outperforms other state-of-the-art methods in both noiseless and noisy cases. 
\begin{enumerate}
\item (Generating an extended support estimate of size bounded by $m$) The extended support is defined by an index set whose size is larger than the actual sparsity $k$ of $x^*$ and includes $\Omega$. gOMP \cite{wang2012generalized}, SP \cite{dai2009subspace}, and CoSaMP \cite{needell2009cosamp} exploited an extended support estimate in the process to estimate $\Omega$. To improve the support recovery performance, the maximum size of the extended support estimate in BMMP is set to $m$ while the maximum size of the extended support estimate in gOMP, SP, and CoSaMP, is $t \cdot \min(k,\left \lfloor {m/t} \right \rfloor)$ , $2k$, and $3k$, respectively,  for a constant $t$ smaller than $k$. The reason for setting this to $m$ in BMMP is explained in the section describing BMMP. 
\item (Generating multiple support candidates) BMMP utilizes the diversity gain by iteratively selecting multiple indices of different size to generate each of multiple support candidates.  In other words, BMMP generates the $i$th support candidate by iteratively selecting $i$ indices through the proposed partial support estimation. There have been related studies generating multiple support candidates through a tree structure: multipath matching pursuit \cite{kwon2014multipath}, multi-branch matching pursuit \cite{rossi2012spatial}. However, the complexity of the tree-based methods scales in the sparsity $k$ and is higher than those of other CS algorithms, since the number of support candidates increases exponentially with $k$ (the depth of the tree) due to the structural nature of the tree. In contrast with these methods, BMMP does not use the tree-based approach  to reduce complexity so that the number $g$ of support candidates is independent of $k$ and the complexity is linearly increased with $g$ in BMMP.
\item (Updating extended support estimate by replacing its subset) To improve the support recovery performance, SP \cite{dai2009subspace} and CoSaMP \cite{needell2009cosamp} update an extended support estimate $\Delta$ by selecting $k$ indices and replacing the rest of the selected indices in $\Delta$ at each iteration. BMMP modifies this technique by iteratively selecting $k/2$ indices in a given extended support estimate $\Delta$  and  updating its complement in $\Delta$.
\item A ridge regression proposed by Wipf and Rao\cite{wipf2004sparse} is used to improve robustness to noise. 
\end{enumerate}

\section{Notation}\label{sec_not}
The set $\{1,2,...,i\}$ is denoted by $\{1:i\}$. The submatrix of a matrix $A:=[a_1,...,a_n]$ $ \in \mathbb{R}^{m \times n}$, where $a_i$ is its $i$th column, with columns indexed by $J \subseteq \{1:n\}$ is denoted by $A_J$. $A^{Q}$ denotes the submatrix of $A$ with rows indexed by $Q \subseteq [m]$. $\mathcal{R}(A)$ denotes the range space spanned by the columns of $A$. $A^{\top}$ denotes the transpose of $A$. $P_{\mathcal{R}(A)}:=A(A^{\top}A)^{-1}A^{\top}$ and $P^{\perp}_{\mathcal{R}(A)}$ denote the orthogonal projection onto $\mathcal{R}(A)$ and its orthogonal complement, respectively. $\left \| \cdot \right \|$ denotes the Frobenius norm.
For a set $\Delta \subseteq \{1:n\}$ and a space $\mathcal{R}(A_{\Delta})$ of $\mathbb{	R}^d$, $\dot a_i({\Delta}):= {P^{\perp}_{\mathcal{R}(A_{\Delta})}a_i}/{\left \| P^{\perp}_{\mathcal{R}(A_{\Delta})}a_i \right \|}$, and $\dot a_i:= {a_i}/{\left \| a_i \right \|}$. Similarly, $\bar a_i({\Delta}):= P^{\perp}_{\mathcal{R}(A_{\Delta})}a_i$.

\section{Partial support estimation by MAP}
\label{sec_MAP}
The proposed MAP-based algorithm for estimating a partial true support is introduced in this section. We assume the following. $\Omega$ is uniformly distributed on $\{1:n\}$. Each nonzero element $u$ of $x^*$ is i.i.d and follows an arbitrary distribution $f_{x^*}(u)$ whose mean and variance are $m_x$ and $\sigma_x^2$, respectively. Each element in $\Phi$ and $w$ follows the Gaussian distribution whose mean is $0$ and variance is $\sigma^2$ and $\sigma_w^2$, respectively.

Suppose that $\Delta \subseteq \{1:n\}$ such that $|\Delta|=:d < m$ is a given partial support estimate.
The goal of the proposed partial support detection is to find an index set belonging to $\Omega \setminus \Delta$ by using the inner product used in RA-ORMP, $g_i(\Delta):= r(\Delta)^{\top} \dot \phi_i(\Delta)$ for $i \in \{1:n\}\setminus \Delta$, where $r(\Delta):=P^{\perp}_{\mathcal{R}(\Phi_{\Delta})}y$ is the residual vector, and $\dot \phi_i(\Delta):= {P}^{\perp}_{\mathcal{R}(\Phi_{\Delta})}\phi_i / \left \| {P}^{\perp}_{\mathcal{R}(\Phi_{\Delta})}\phi_i\right \|$. $g_i(\Delta)$ represents the correlation between the residual vector and $i$.

Although MAP support detection \cite{lee2016map} estimates a partial true support from the likelihood ratio of the inner product $h_i(\Delta):=   r(\Delta)^{\top}\dot \phi_i $, where $\dot \phi_i:= \phi_i / \left \| \phi_i\right \|$, the proposed method uses the likelihood ratio of $g_i(\Delta)$ for $i \in \{1:n\}\setminus \Delta$. Note that $g_i(\Delta)$ and $h_i(\Delta)$ use $\dot\phi_i(\Delta)$ and $\dot\phi_i$, respectively. {$h_i(\Delta)$ can be interpreted as the correlation term in a normalized version of OMP since the normalized vector $\dot \phi_i$ is used in $h_i(\Delta)$, whereas OMP uses $\phi_i$ in its correlation term $({P}^{\perp}_{\mathcal{R}(\Phi_{\Delta})}y)^{\top} \phi_i $.}
This indicates that the proposed method additionally considers the orthogonal complement ${P}^{\perp}_{\mathcal{R}(\Phi_{\Delta})}$ of $\Phi_{\Delta}$ for  the column vector $\phi_{i}$ to represent the correlation compared to MAP support detection. Since each vector space $\mathcal{R}(\dot\phi_i(\Delta))$ for $i \in \Omega \setminus \Delta$ is guaranteed to belong to a space $\cup_{j \in \Omega \setminus \Delta} \mathcal{R}({P}^{\perp}_{\mathcal{R}(\Phi_{\Delta})}\phi_j)$, which is extended from the residual space $ \mathcal{R}({P}^{\perp}_{\mathcal{R}(\Phi_{\Delta})}y)$, the index selection based on $g_i(\Delta)$ may more successfully find the index in $\Omega \setminus \Delta$ compared to that based on $h_i(\Delta)$.\footnote{It is not guaranteed that each vector space $\mathcal{R}(h_i(\Delta))$ for $i \in \Omega \setminus \Delta$ belongs to the extended residual space $\cup_{j \in \Omega \setminus \Delta} \mathcal{R}({P}^{\perp}_{\mathcal{R}(\Phi_{\Delta})}\phi_j)$.}

The proposed partial support detection uses an estimate of $|\Omega \setminus \Delta|$. This is obtained by the following derivations.
\begin{align}\nonumber
\mathbb{E}[\left \| r(\Delta) \right \|^2]
&\overset{(a)}= \sum\limits_{l \in \Omega \setminus \Delta} \mathbb{E}[ \left \| \bar \phi_l(\Delta) \right \|^2] v_x^2+ \mathbb{E}[\left \| \bar w(\Delta) \right \|^2 ] \\\nonumber
&\overset{(b)}{=}(|\Omega \setminus \Delta| \cdot \sigma^2 \cdot v_x^2 + \sigma_w^2) \cdot (m-d) ,
\end{align}
where $v_x:=\sqrt{(m^2_x+ \sigma_x^2)}$, (a) follows from the independence assumption, and (b) follows from Lemma \ref{Lemma1} (with $A=\Phi_{\Delta}$ and $b=\phi_l$ or $w$) in that $\mathbb{E}[ \left \| \bar \phi_l(\Delta) \right \|^2]= \sigma^2 \cdot (m-d)$ and $\mathbb{E}[\left \| \bar w(\Delta) \right \|^2 ]= \sigma_w^2 \cdot (m-d)$.
Then, the estimate $\Psi(\Delta)$ for $|\Omega \setminus \Delta|$ is given as
\begin{align}\label{ers}
\Psi(\Delta) :=\frac{\max(\left \| r(\Delta) \right \|^2 /  (m-d) - \sigma_w^2,0) }{  \sigma ^2  \cdot v_x^2  }.
\end{align}
Note that the following equalities hold for $i \in \{1:n\}\setminus \Delta$
\begin{align}\nonumber
\footnotesize
g_i(\Delta) &\overset{(a)}{=} \dot \phi_i(\Delta)^{\top} (\sum\limits_{l \in \Omega \setminus \Delta} \bar \phi_l(\Delta) x_l+ \bar w(\Delta)) \\\label{inva}
&= \left\| \bar \phi_i(\Delta) \right\| x_i + \sum\limits_{l \in \Omega \setminus ( \Delta\cup \{i\})} \dot \phi_i(\Delta)^{\top} \bar \phi_l(\Delta) x_l + \dot \phi_i(\Delta)^{\top} \bar w(\Delta),
\end{align}
where $x_i$ is the $i$th element of $x^*$ and (a) follows from $r(\Delta) = \sum\limits_{l \in \Omega \setminus \Delta} \bar \phi_l(\Delta) x_l+ \bar w(\Delta)$.
Our derivation in the rest of this section is based on \cite{lee2016map}.  From the equalities in (\ref{inva}), $g_i(\Delta)$ becomes the right-side terms in (\ref{hypo0}) and (\ref{hypo1}) under the null and true hypotheses---$\mathcal{T}_0$ and $\mathcal{T}_1$, respectively. For $\mathcal{T}_0$, $i$ does not belong to $\Omega$, and $x_i = 0$. For $\mathcal{T}_1$, $i$ belongs to $\Omega$, and $x_i=u$.

\begin{align}\label{hypo0}
\mathcal{T}_0 :g_i(\Delta) & = \sum\limits_{l \in \Omega \setminus\Delta} \dot \phi_i(\Delta)^{\top} \bar \phi_l(\Delta) x_l + \dot \phi_i(\Delta)^{\top} \bar w(\Delta)\\\label{hypo1}
\mathcal{T}_1 :g_i(\Delta) & = \left\| \bar \phi_i(\Delta) \right\| u+\sum\limits_{l \in \Omega \setminus ( \Delta\cup \{i\})} \dot \phi_i(\Delta)^{\top} \bar \phi_l(\Delta) x_l +\dot \phi_i(\Delta)^{\top} \bar w(\Delta)
\end{align}

From (\ref{inva}) and (\ref{hypo0}), the expectation and variance of $z_i:=g_i(\Delta)$ under the null hypothesis can be respectively approximated as follows:
\begin{align}\nonumber
\mathbb{E}[z_i|x_i=0]&= \sum\limits_{l \in \Omega \setminus\Delta} \mathbb{E}[(\dot \phi_i(\Delta)^{\top} \bar \phi_l(\Delta))] m_x+ \mathbb{E}[(\dot \phi_i(\Delta)^{\top} \bar w(\Delta))], \\\label{exp1}
&=\sum\limits_{l \in \Omega \setminus\Delta} \mathbb{E}[(\dot \phi_i(\Delta)^{\top} \phi_l)] m_x+ \mathbb{E}[(\dot \phi_i(\Delta)^{\top} w)]\overset{(a)}{=}0
\end{align}
\begin{align}\nonumber
\mathbb{E}[z_i^2|x_i=0] &= \sum\limits_{l \in \Omega \setminus\Delta} \mathbb{E}[(\dot \phi_i(\Delta)^{\top} \bar \phi_l(\Delta))^2] v_x^2+ \mathbb{E}[(\dot \phi_i(\Delta)^{\top} \bar w(\Delta))^2]\\\nonumber
&=\sum\limits_{l \in \Omega \setminus\Delta} \mathbb{E}[(\dot \phi_i(\Delta)^{\top}\phi_l)^2] v_x^2+ \mathbb{E}[(\dot \phi_i(\Delta)^{\top}w)^2]\\\nonumber
&\overset{(a)}{=}|\Omega \setminus\Delta|\cdot \sigma^2 v_x^2+\sigma_w^2 \\\label{var1}
&\overset{(b)}{\simeq}\Psi(\Delta) \cdot \sigma^2 v_x^2+\sigma_w^2 =:\tilde \sigma^2_0,
\end{align}
where (a) and (b) follow from the circular symmetry of the Gaussian distribution and (\ref{ers}), respectively. Similarly, by (\ref{hypo1}) and Lemma \ref{Lemma1} (with $A=\Phi_{\Delta}$ and $b=\phi_i$), the expectation and variance of $z_i$ under the true hypothesis are obtained respectively as
\begin{align}\label{exp2}
&\mathbb{E}[z_i|x_i=u] = \mathbb{E}[\left\| \bar \phi_i(\Delta) \right\|_2 u ]={\sigma \tau} u
\end{align}
and
\begin{align}\nonumber
\mathbb{E}[(z_i-\mathbb{E}[z_i])^2|x_i=u] &\simeq\sum\limits_{l \in \Omega \setminus (\Delta \cup \{i\}) } \mathbb{E}[(\dot \phi_i(\Delta)^{\top} \bar \phi_l(\Delta))^2] v_x^2+ \mathbb{E}[(\dot \phi_i(\Delta)^{\top} \bar w(\Delta) )^2]\\\label{var2}
&\simeq(\Psi(\Delta)-1) \cdot \sigma^2 v_x^2+\sigma_w^2 =:\tilde \sigma^2_1,
\end{align}
where $\tau:={\sqrt{2} \cdot \Gamma(\frac{1+m-d}{2})}/{\Gamma(\frac{m-d}{2})}$, and $\Gamma(\cdot)$ is the Gamma function.

Then, using (\ref{exp1})--(\ref{var2}), we obtain the log-likelihood ratio $\Theta(z_i) := \ln ({\mathbb{P}(i \in \Omega | z_i)}/{\mathbb{P}(i \notin \Omega | z_i)})$ as follows.
\begin{align}\nonumber
\Theta(z_i) &\propto \ln (\frac{\mathbb{P}(z_i|i \in \Omega)}{\mathbb{P}(z_i|i \notin \Omega)})\\\nonumber
&= \ln (\frac{\int_{-\infty}^{\infty}\mathbb{P}(z_i|x_i=u)f_{x^*}(u)du}{\mathbb{P}(z_i|x_i=0)})\\\label{loglikelihood}
&\overset{(a)}{\simeq} { \ln (\frac{\int_{-\infty}^{\infty}\frac{1}{\tilde \sigma_1 \sqrt{2 \pi}} \exp (- \frac{|z_i-{\sigma \tau} u|^2}{2 \tilde \sigma_1^2})f_{x^*}(u)du}{\frac{1}{\tilde \sigma_0 \sqrt{2 \pi}} \exp (- \frac{|z_i|^2}{2 \tilde \sigma_0^2})}) },
\end{align}
where (a) follows from the Gaussian approximation of the conditional distributions of $z_i$ given $x_i=0$ ($\mathcal{T}_0$) by using (\ref{exp1}) and (\ref{var1}) and given $x_i=u$ ($\mathcal{T}_1$) by using (\ref{exp2}) and (\ref{var2}).

Therefore, the proposed method estimates a partial true support as an index set $\Lambda$ by selecting the $|\Lambda|$ largest ratios $\Theta(z_i)$ for $i  \in \{1:n\} \setminus \Delta$ in $(\ref{loglikelihood})$. Note that $\Theta(z_i)$ is propotional to (\ref{L}) when each nonzero element of $x^*$ follows the uniform distribution $f_{x^*}(u;a,b) = 1/(b-a)$ for $a \leq u\leq b$
\begin{align}\label{L}
\frac{(z_i)^2}{2 \tilde\sigma_0^2} +\ln (\textup{erf}(\frac{{\sigma \tau} b - z_i}{\tilde \sigma_1 \sqrt{2}})-\textup{erf}(\frac{{\sigma \tau} a - z_i}{\tilde \sigma_1 \sqrt{2}})),
\end{align}
where $\textup{erf}(x):= \frac{2}{\sqrt{\pi}} \int_0^x e^{-t^2}dt$.

\begin{section}{Algorithm description}\label{sec_algd}
\begin{algorithm}[t]
  \caption{BMMP($y,\Phi,k,g,\epsilon, \lambda,\boldsymbol{p}$)}
    \begin{algorithmic}[1]
	\Input{$y \in \mathbb{R}^{m}, \Phi \in \mathbb{R}^{m \times n}$, $(g,k)\in \mathbb{N}^2$, $				\boldsymbol{p}:=(\sigma,m_x,\sigma_x,\sigma_w) \in \mathbb{R}^4$, $\epsilon \in \mathbb{R}$}
	\Output{$\hat x \in \mathbb{R}^{n}, \hat \Omega \subseteq \{1:n\}$}
	\For{ $t=1 \textup{ to } g$}
	\State $(\Delta,i) \gets (\{\}, 1)$, where $\{\}$ is the empty set.
	 \While{$i=1$ or $\left \| P^{\perp}_{\mathcal{R}(\Phi_{\bar \Omega_{i}})}y \right \| < \left \| P^{\perp}_{\mathcal{R}(\Phi_{\bar \Omega_{i-1}})}y \right \| $}
	\While{ $|\Delta| \neq m$ and $\left \| P^{\perp}_{\mathcal{R}(\Phi_{\Delta})}y \right\| > \epsilon$}
	\State $z_q \gets ({P}^{\perp}_{\mathcal{R}(\Phi_{\Delta})}y)^{\top} \dot \phi_q(\Delta) \textup{ for } q \in \Delta^c := \{1:n\} \setminus \Delta$
	\State $\bar z_q \gets  \Theta(z_q;y,\Delta,\sigma,m_x,\sigma_x,\sigma_w)  \textup{ for }q \in \Delta^c $
	\State $v \gets \min(t,m-|\Delta|)$
	\State $\Lambda \gets \{ \textup{indices $q$ of the $v$-largest values of $\bar z_q$ in $\Delta^c $} \}$
	\State $\Delta \gets \Delta \cup \Lambda$
	\EndWhile
	\State $\bar{x}^{{\Delta}}  \gets (\Phi^{\top}_{{\Delta}}\Phi_{{\Delta}}+ \eta^2 D({\gamma})^{-1})^{-1}\Phi^{\top}_{{\Delta}} y  $  where $({\gamma},\eta)$ is obtained by $\underset{{{\bar \gamma} \in \mathbb{R}^{a},\bar\eta \in \mathbb{R}}}{\arg\min\limits} \, L(\Phi_{{\Delta}},y)$
	\For{$q \in \Delta $}
	\State   $\zeta^i_q \gets | \bar{x}_q| $,  where $\bar{x} = (\bar{x}_1,...,\bar{x}_n)^{\top}$
	\EndFor
	\State $\bar \Omega_{i+1} \gets \{ \textup{indices $z$ of the $k$-largest values of ($\zeta^i_z$) in $\Delta$}\}$
	\State $\Delta \gets \{ \textup{indices $z$ of the $\left \lfloor {k/2} \right \rfloor$-largest values of ($\zeta^i_z$) in $\Delta$}\}$
	\State{$i \gets i +1 $}
	\EndWhile
	\State{$ \dot\Omega_t \gets \bar \Omega_{i-1}$}
	\If{$\left \| P^{\perp}_{\mathcal{R}(\Phi_{\dot\Omega_t })} y \right \| \leq \epsilon  $} go to step 23
	\EndIf
	\EndFor
	\State $p \gets \underset{b \in \{1:t\}}{\arg \min} \,\, \left \| P^{\perp}_{\mathcal{R}(\Phi_{\dot \Omega_{b}})} y \right \|$
	\State \Return{$ (\hat\Omega:={\dot\Omega}_{p}, \hat x:= \underset{x \in \mathbb{R}^n}{\arg\min}\left \| \Phi_{{\dot\Omega}_{p}} x^{{\dot\Omega}_{p}} - y \right \|)$}
  \end{algorithmic}
\label{alg:BMMP}
\end{algorithm}
\begin{figure}[t]
\centering
\subfigure{\includegraphics[width=8.5cm,height=6.6cm]{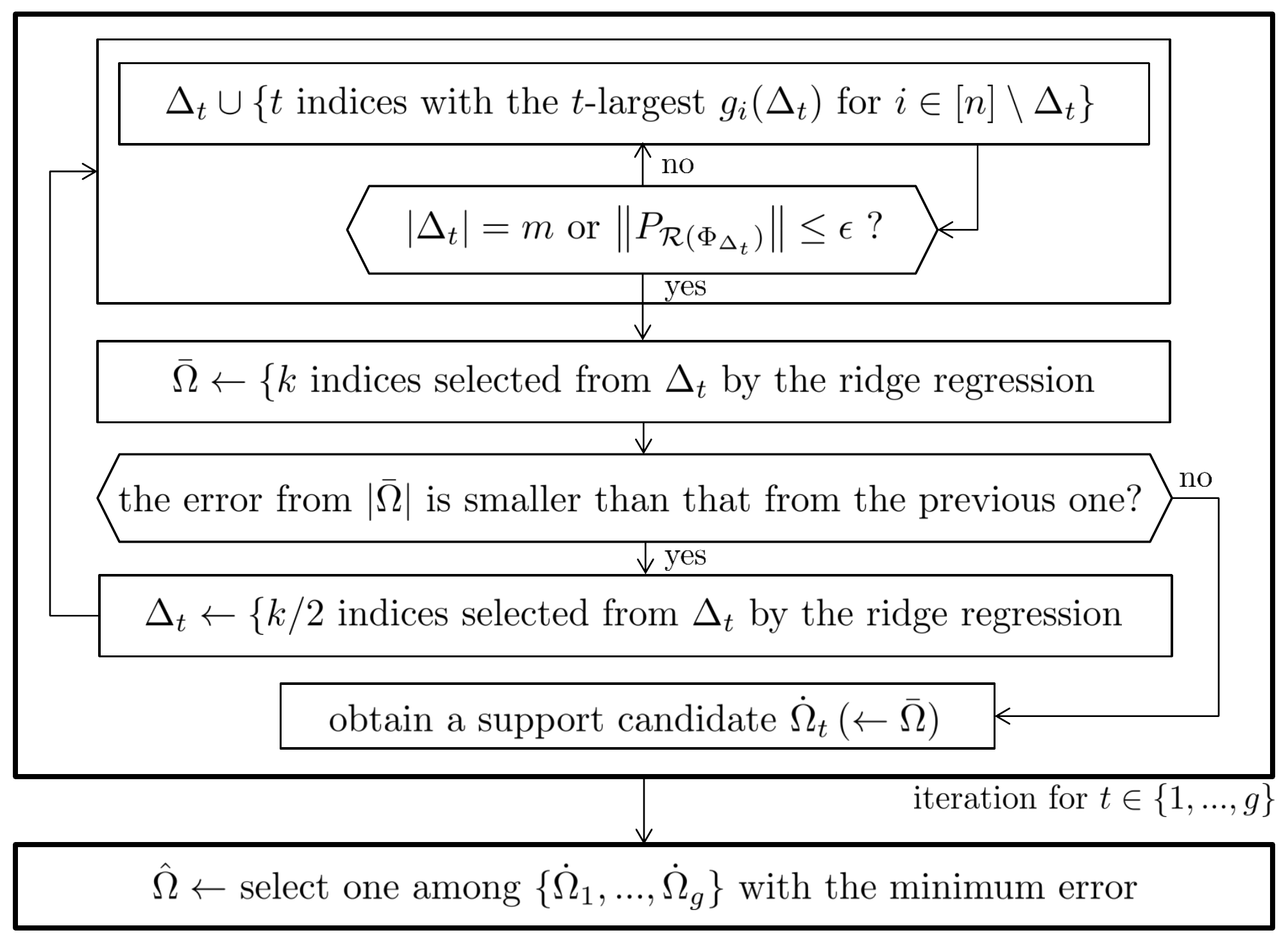}}
\caption{Description of BMMP:}  {$\hat \Omega$ is the support estimate obtained from BMMP and $g$ is the number of multiple support candidates among which $\hat \Omega$ is selected. We note that the $t$th support candidate $\bar \Omega_t$ is obtained by selecting $k$ incides from an extended support estimate $\Delta_t$ whose size is smaller than or equal to $m$, which is generated by iteratively selecting $t$ indices of the $t$-largest log-likelihood ratios $\Theta(g_i(\Delta_t))$ for $i\in \{1:n\}\setminus \Delta_t$ and adding them to $\Delta_t$.}
\label{Figure1}
\end{figure}

The proposed method, BMMP (Algorithm \ref{alg:BMMP}), for estimating ${x^*}$ and $\Omega$ is introduced in this section. BMMP returns ($\hat x, \hat \Omega$) as estimates of (${x^*},\Omega$) given the tuple ($y,\Phi,k,g,\epsilon$), where $g$ is the number of support candidates from which the final support is estimated, and $\epsilon$ is a threshold used to terminate BMMP when the residual error is smaller than this value. The estimate $\hat \Omega$ is obtained in Steps 23--24 by selecting one of $g$ support candidates ($\dot \Omega_{t}$ for $t \in \{1:g\}$) with the minimum residual error. The $t$th support candidate $\dot \Omega_{t}$ in Step 23 is obtained by the iteration in Steps $3$--$18$ consisting of the following four stages.

In the first stage (generate an extended support estimate $\Delta$) in Steps 4--10, a set $\Lambda$ of $t$ indices in Step 8 is iteratively obtained through the proposed partial support estimation\footnote{We note that $\Theta(z_i)$ in (\ref{L}) is calculated by using the following values: $(y,\Delta, \sigma, $ $m_x, \sigma_x, \sigma_w)$. Step 6 in Algorithm \ref{alg:BMMP} specifies these values as inputs of $\Theta(\cdot)$.} and added to a set $\Delta$ in Step 9 until its size $|\Delta|$ is equal to $m$ or the residual error obtained from the estimate is smaller than $\epsilon$, i.e.  $\left \| P^{\perp}_{\mathcal{R}(\Phi_{\Delta})}y \right\| \leq  \epsilon$. In the second stage (generate a temporary estimate $\bar x$ of ${x^*}$) in Step 11, an estimate $\bar x$ of the signal vector supported on $\Delta$ is obtained by the ridge regression described in Step 11. In the third stage in Steps 12--15, a temporary support estimate $\bar\Omega_{i+1}$ is obtained by selecting the $k$-largest elements in $\bar x^{\Delta}$. The fourth stage in Steps 16--17 generates a subset of $\Delta$ by selecting $k/2$ indices from $\Delta$. Then, the extended support estimate $\Delta$ is reinitialized as the subset in the next iteration step. This iteration terminates and the $t$th support candidate $\dot \Omega_{t}$ is obtained in Step 19 when the residual error obtained from the temporary support estimate $\bar \Omega_{i+1}$ in Step 15 does not decrease further in Step 3.

To enhance the robustness to noise, a ridge regression proposed by Wipf and Rao \cite{wipf2004sparse} is used in Step 11 with the following parameters $(\gamma,\eta)$ obtained by minimizing the cost $L(\cdot)$ in (\ref{sble}). 
\begin{align}\label{sble}
(\gamma,\eta) = \underset{{\bar \gamma \in \mathbb{R}^{a},\bar \eta \in \mathbb{R}}}{\arg\min\limits} \,L(\bar \Phi,y):=\, \log |\Sigma|+ y^{\top}\Sigma^{-1}y, 
\end{align}
where $\bar \Phi := \Phi_{\Delta} \in \mathbb{R}^{m \times a}$ is a submatrix of $\Phi$ with columns indexed by $\Delta$, $\Sigma = (\bar \eta^{-2} {\bar\Phi}^{\top}{\bar\Phi}+ D(\bar \gamma))^{-1}$, and $D(\bar \gamma)$ is the diagonal matrix whose $i$th diagonal element is $\gamma_i$ in $\bar \gamma=(\gamma_1,...,\gamma_a)^{\top}$. We obtain an appoximated solution of $(\ref{sble})$ by using sparse bayesian learning (SBL) \cite{wipf2004sparse} to reduce the complexity. 

In the noiseless case, the parameters $(\bar \gamma,\bar \eta)$ in the ridge regression are set to zero so that the ridge regression becomes the least-square regression ($\underset{x \in \mathbb{R}^n}{\arg\min}\left \| \Phi_{\Delta} x^{\Delta } - y \right \|$).  The least-squares method has a lower complexity than SBL and it is well-known that the least-squares method provides a unique solution as ${x^*}$ in the noiseless case if $|\Delta| \leq m$, $\Delta \supseteq \Omega$, and $\Phi_{\Delta}$ has the full column rank. If $\Delta \supseteq \Omega$ holds, the inversion problem of CS can be simplified as an easier problem where $\Phi$ is replaced by its submatrix $\Phi_{\Delta}$ compared to the original problem. When $\Delta$ has a larger size, the probability of satisfying $\Delta \supseteq \Omega$ is greater. For this reason, we set the maximal size of the extended support estimate $\Delta$ as $m$. Besides, it is guaranteed from Lemma \ref{Lemma2} that in the noiseless case, $\Omega$ is a subset of $\Delta$ ($\Delta \supseteq \Omega$) if $|\Delta| < m$ and $\left \| P^{\perp}_{\mathcal{R}(\Phi_{\Delta})}y \right \|$ is equal to zero. Based on this fact, to minimize the size of $\Delta$ satisfying $\Delta \supseteq \Omega$, we increase the size of $\Delta$ from $0$ to $m$ and find $\Delta$ satisfying the criterion $\left \| P^{\perp}_{\mathcal{R}(\Phi_{\Delta})}y \right\| \leq \epsilon$ in Steps 4--10. Figure \ref{Figure1}  illustrates the procedure of BMMP.
\end{section}

\begin{figure}[t]
  \centering
     \subfigure[Support recovery rate\label{Figure2a}]{\includegraphics[width=7.5cm,height=4.2cm]{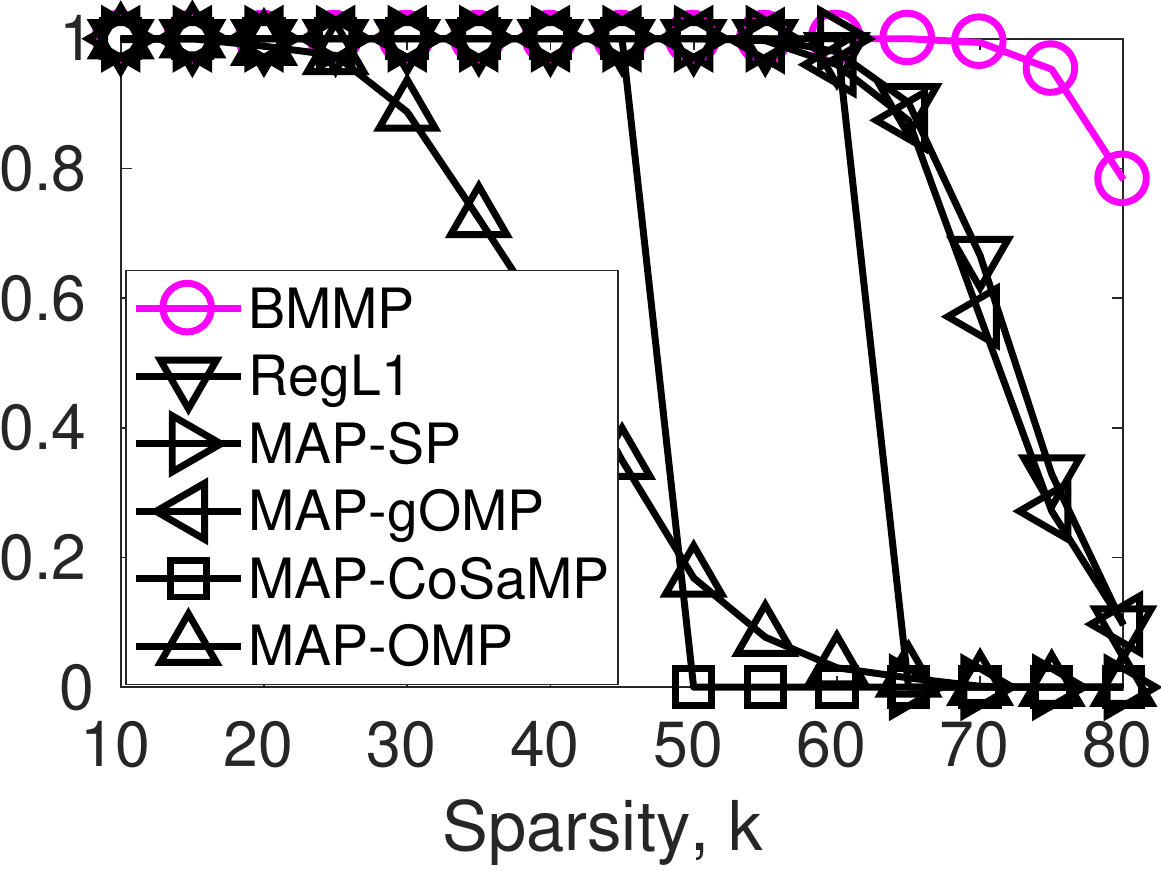}}
     \subfigure[Execution time\label{Figure2b}]{\includegraphics[width=7.5cm,height=4.2cm]{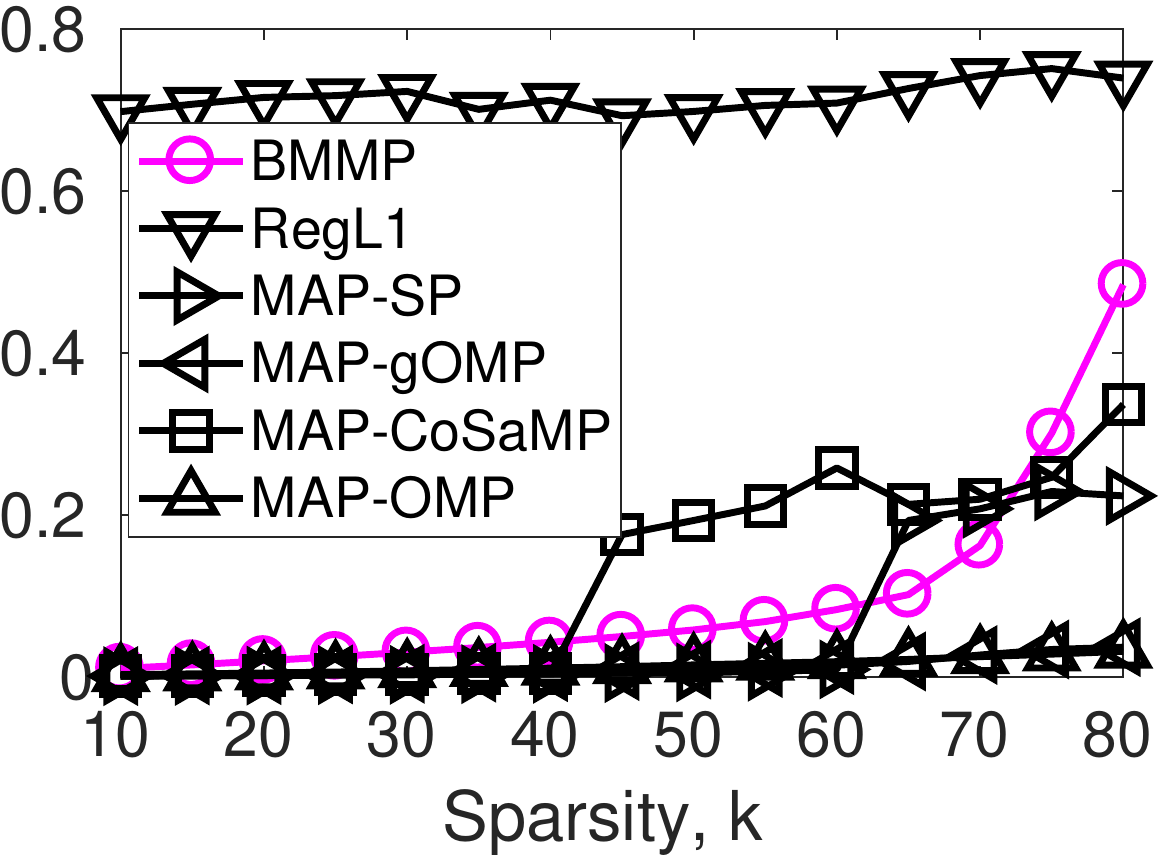}}
     \subfigure[Mean squared error\label{Figure2c}]{\includegraphics[width=7.8cm,height=4.2cm]{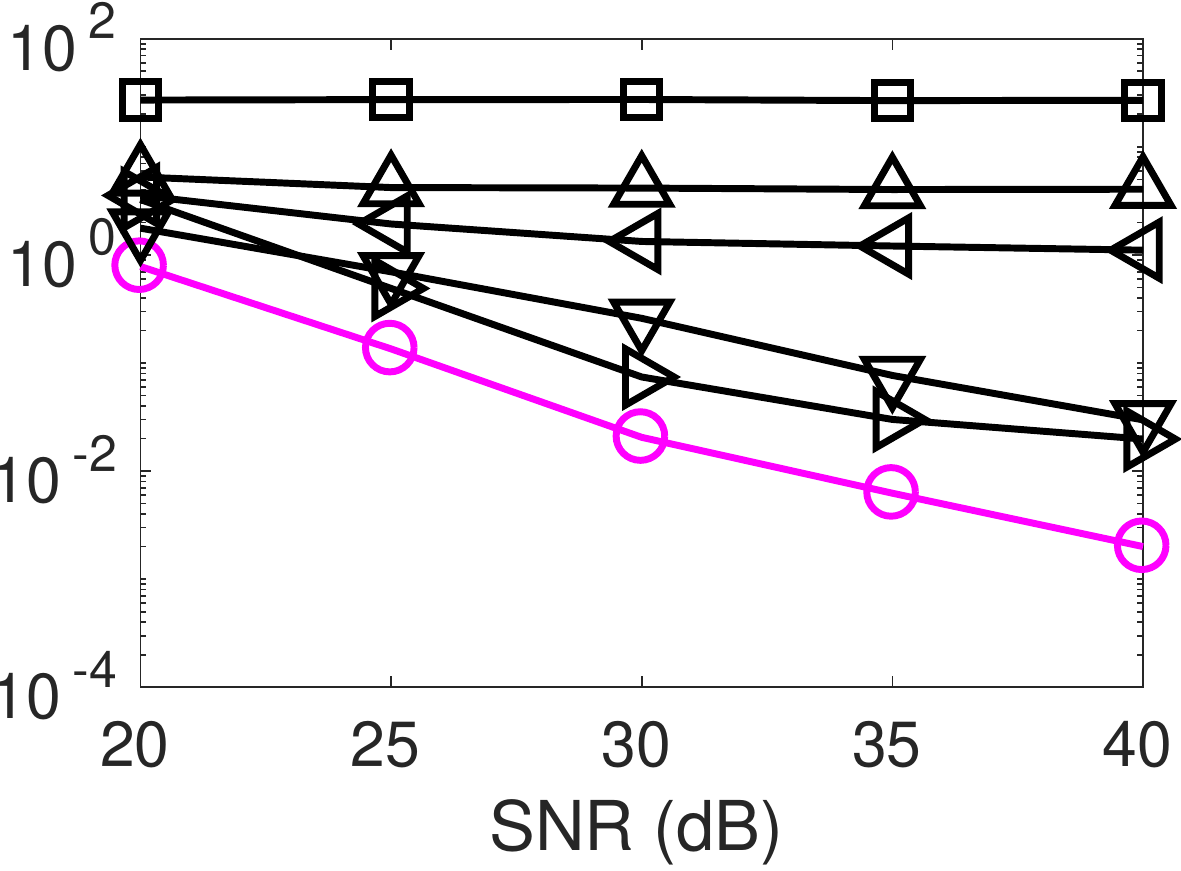}}
     \subfigure[Support recovery rate\label{Figure2d}]{\includegraphics[width=7.7cm,height=4.2cm]{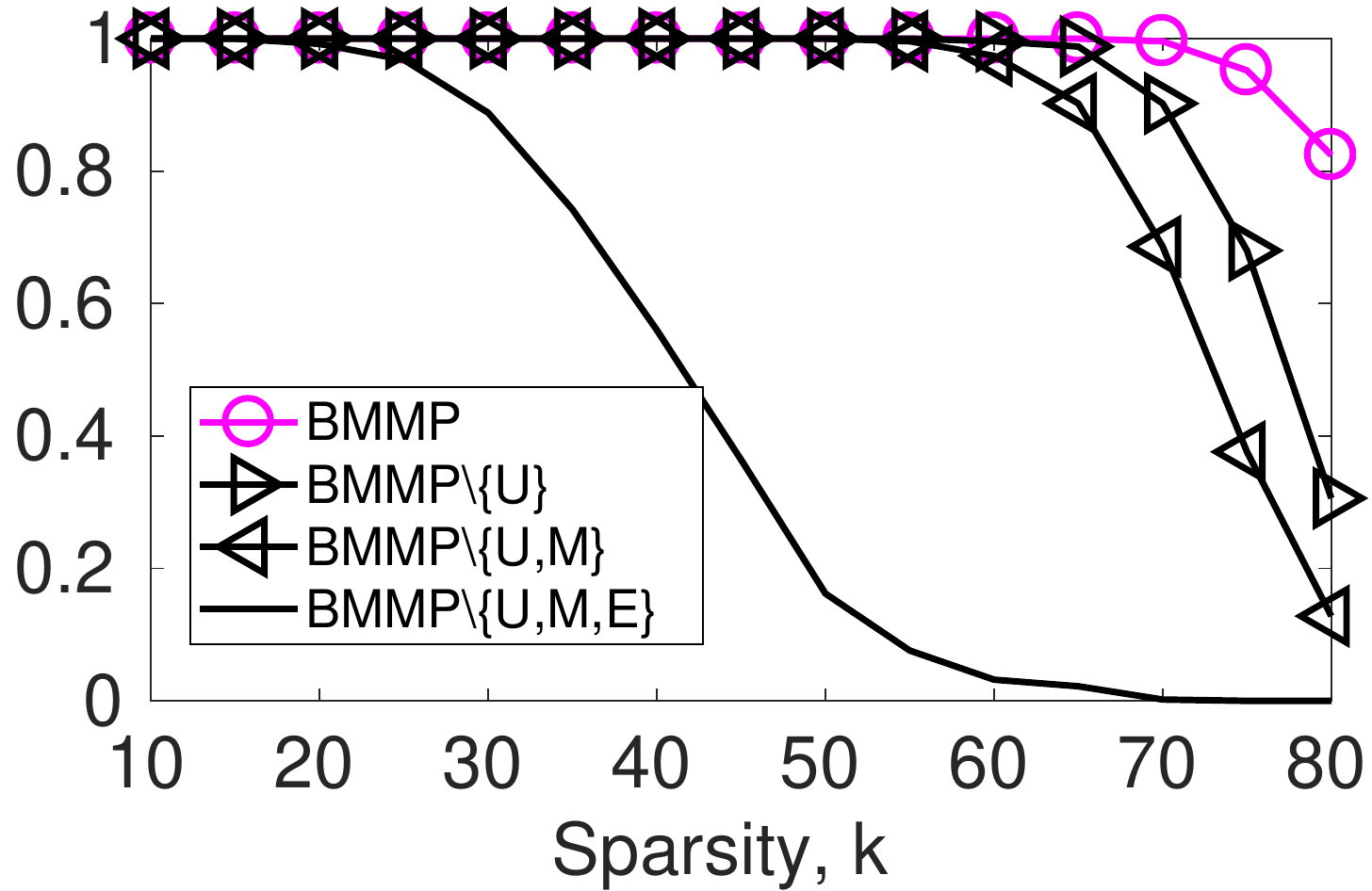}}
  \caption{Performance comparison of BMMP and related algorithms}
  \label{Figure2}
\end{figure}

\begin{figure} 
  \centering
     \subfigure[Support recovery rate]{\includegraphics[width=7.8cm,height=4.2cm]{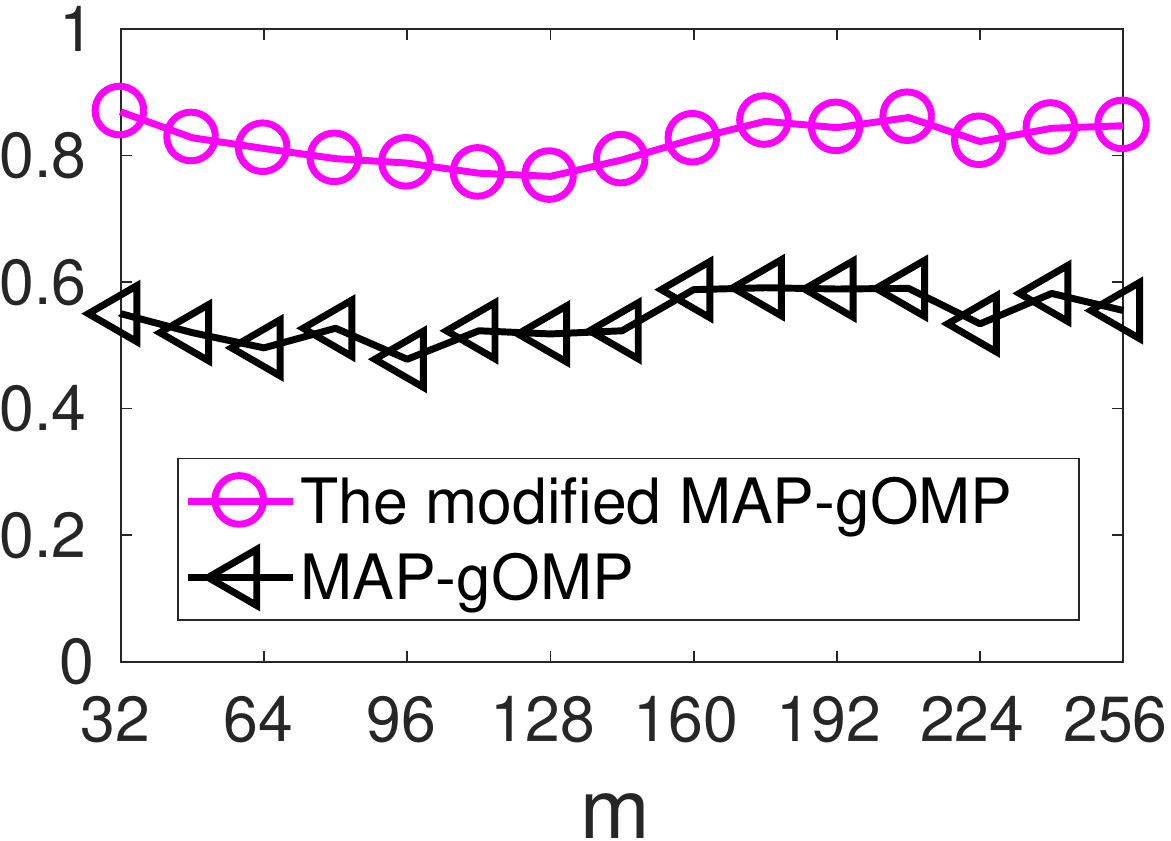}}
  \caption{Performance comparison of the proposed partial support detection and the MAP support detection}
  \label{Figure3}
\end{figure}

\section{Numerical experiments}
\label{sec_ne}

We compare the performance of BMMP and existing state-of-the-art algorithms such as MAP-gOMP\footnote{Two indices are selected at each step in MAP-gOMP.}, MAP-SP, MAP-CoSaMP, MAP-OMP \cite{lee2016map}, RegL1 ($l_{1}$ norm minimization with a nonnegative constraint).\footnote{RegL1 outputs the signal estimate $\hat x=(x_1,...,x_n)^{\top}$ by minimizing $\sum_i |x_i|$ such that $x_i \geq 0$ for $i \in \{1:n\}$ and $\left \| y-\Phi \hat x \right \| \leq \tau$ where $\tau$ is the noise magnitude.} We note that MAP-SP, MAP-CoSaMP, MAP-OMP are algorithms based on MAP support detection \cite{lee2016map}. Each entry of $\Phi$ is i.i.d. and follows the Gaussian distribution $\mathcal{N}(0,1/m)$ with mean zero and variance $1/m$. The elements of ${(x^*)}^{\Omega}$ are independently and uniformly  sampled from $0$ to $1$.
Then, $(\sigma,m_x,\sigma_x)$ is $(1/\sqrt{m},1/2,1/\sqrt{12})$.

\begin{figure*}[t]
  \centering
  \footnotesize
\subfigure{\includegraphics[width=18cm]{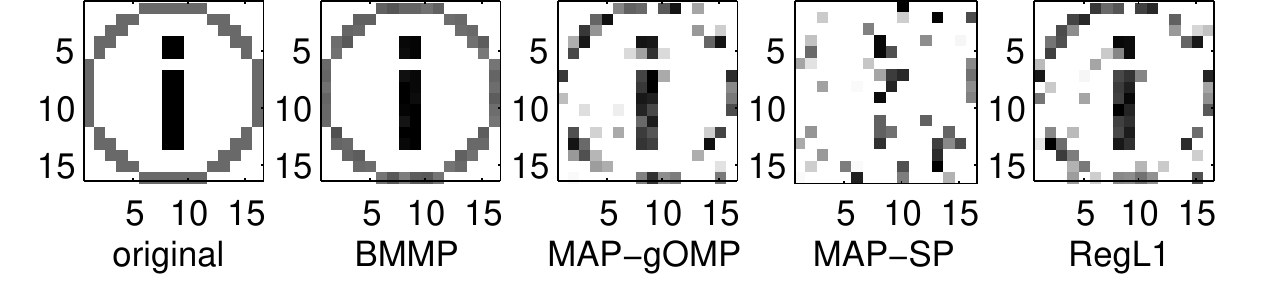}}
\caption{Reconstruction performance comparison of a sparse image with compressed and noisy measurements.}
\label{Figure4}
\end{figure*}

The rate of successful support recovery, i.e., $\hat \Omega = \Omega$, and the execution time of each algorithms in the noiseless case are shown in Fig. \ref{Figure2a} and Fig. \ref{Figure2b}, respectively. ($m,n$) is set to ($128,256$) in Fig. \ref{Figure2}. We evaluate BMMP whose input $(g,\epsilon,\lambda)$ is set to $(4,\left \| y \right \| \cdot 10^{-\textup{SNR}_{\textup{dB}}/20},\epsilon^2/m)$, where SNR$_{\textup{dB}}$ is the signal-to-noise ratio (SNR $:= \mathbb{E}\left \| \Phi_{\Omega} x^{\Omega}_0 \right \|^2  / \, \mathbb{E}\left \|w \right \|^2$) in decibels. {It has been shown \cite{foucart2013mathematical} that in the noiseless case, the maximal sparsity $\bar k$ required for an ideal approach to recover any signal vector ${x^*}$ such that $|\Omega| \leq \bar k$ is equal to $m/2$ given $m$, where $\Omega$ is the support of ${x^*}$; the $l_0$  bound is denoted by this ideal limit of the signal sparsity ($m/2$).}  Fig. \ref{Figure2a} and Fig. \ref{Figure2b} show that BMMP has a lower complexity than RegL1, outperforms other state-of-the-art methods, and approaches the $l_0$  bound, i.e., $m/2=64$.

Fig. \ref{Figure2c} illustrates the mean squared error $\mathbb{E}\left \| \hat x - {x^*} \right \|^2$ of each algorithm when $\textup{SNR}$ is varied from $20$ to $40$ dB\footnote{In the noisy case, $x^*$ is randomly sampled such that each of its non-zero elements is ranged from $0.1$ to $1$.}, the sparsity $k$ is fixed to $60$. It is observed that the signal reconstruction performance of BMMP is better than those of other methods in the noisy case.

Fig. \ref{Figure2d} shows how each of the following three techniques used in BMMP contributes to the performance improvement of BMMP: exploiting the extended support estimate with its maximal size equal to $m$ ($\mathcal{E}$), using multiple support estimates ($\mathcal{M}$), and iteratively updating the support estimate by replacing its subset ($\mathcal{U}$).  BMMP$\setminus \{\mathcal{U}\}$ refers to BMMP without $\mathcal{U}$ by setting the value $\left \lfloor {k/2} \right \rfloor$ to $0$ in step 16 in Algorithm \ref{alg:BMMP}. BMMP$\setminus \{\mathcal{U},\mathcal{M}\}$ denotes BMMP without $\{\mathcal{U},\mathcal{M}\}$ by setting $g$ to $1$ in BMMP$\setminus \{\mathcal{U}\}$. Similarly, BMMP$\setminus \{\mathcal{U},\mathcal{M},\mathcal{E}\}$ indicates BMMP without $\{\mathcal{U},\mathcal{M},\mathcal{E}\}$ such that $m$ (the maximum size of the extended support estimate $\Delta$) shown in steps 4-10 in Algorithm \ref{alg:BMMP} is set to $k$ and the remaining settings follow BMMP$\setminus \{\mathcal{U},\mathcal{M}\}$.
It is shown that the performance improves in the order of BMMP$\setminus \{\mathcal{U},\mathcal{M},\mathcal{E}\}$, BMMP$\setminus\{\mathcal{U},\mathcal{M}\}$, BMMP$\setminus \{\mathcal{U}\}$, and BMMP. {Fig. \ref{Figure2d} shows that the consideration of $\mathcal{E}$ contributes to a significant performance improvement of BMMP.}

To compare performance of the proposed partial support estimation and MAP support detection \cite{lee2016map}, we compare performance of MAP-gOMP and its variant, which is obtained by replacing the MAP support detection with the proposed partial support estimation in MAP-gOMP. Fig. \ref{Figure3} shows the rate of successful support recovery of these two algorithms in the noiseless case where $m$ is varied from $32$ to $256$ with $k=\left \lfloor {m/1.8} \right \rfloor$ and $n=2m$. It is observed  that the proposed partial support estimation outperforms the MAP support detection \cite{lee2016map}.

Fig. \ref{Figure4} shows the performance for reconstructing a grayscale sparse image with a size of $16 \times 16$ pixels in the noisy case when SNR = $25$ dB. This image is compressed by using  $\Phi \in \mathbb{R}^{138 \times 256}$, whose elements are i.i.d. and sampled from the Gaussian distribution $\mathcal{N}(0,$ $1/138)$ with mean $0$ and variance $1/138$. It is observed in Fig. \ref{Figure4} that BMMP recovers the image better than other methods.

\section{Conclusion} 
We presented BMMP, which updates multiple extended support estimates with each size equal to $m$ and performs a likelihood ratio test given the correlation term in RA-ORMP. The numerical results show that BMMP achieves an improvement in performance compared to existing state-of-the-art methods, even with noisy measurements. Future related studies will consider the application of a technique for estimating the moments of $\Phi$, $w$, and $x^*$ to BMMP and the development of BMMP in the case where the distribution of $x^*$ is designable or the a priori information of $x^*$ is available, i.e., a communication system for nonorthogonal multiple access \cite{wang2016joint,cirik2017multi}.

\appendix
\section{Lemmas}
\begin{lem}\label{Lemma1} Let $b\in \mathbb{R}^{m}$ be a vector such that its elements are i.i.d. and follow the Gaussian distribution $\mathcal{N}(0,\hat \sigma^2)$, and let $A:=[a_1,...,a_{d}] \in \mathbb{R}^{m \times d}$ be a matrix $A$ with rank $d<m$. Then, it is guaranteed that   
$\mathbb{E}[\left \| P^{\perp}_{\mathcal{R}(A)}b \right \|]=\hat \sigma \cdot \tau$ and $\mathbb{E}[\left \| P^{\perp}_{\mathcal{R}(A)}b \right \|^2]=\hat \sigma^2 \cdot (m-d)$, where $\tau:={\sqrt{2} \cdot \Gamma(\frac{1+m-d}{2})}/{\Gamma(\frac{m-d}{2})}$, and $\Gamma(\cdot)$ is the Gamma function.
\end{lem}
\begin{proof}
Let $\{e_1,...,e_{m-d}\}$ be an orthonormal basis of $\mathcal{R}^{\perp}(A)$. Then, $P^{\perp}_{\mathcal{R}(A)}b = \sum_{j=1}^{m-d}(b^{\top} e_j)e_j$.
Note that each $b^{\top}e_{j}$  for $j \in \{1:m-d\}$ follows $\mathcal{N}(0,\hat \sigma^2 )$ by the isotropic property of the Gaussian vector \cite{tse2005fundamentals}. Since it follows that
\begin{align}\nonumber
\left \| P^{\perp}_{\mathcal{R}(A)}b \right \| = \sqrt{\sum_{j=1}^{m-d}(b^{\top} e_j)^{2}},
\end{align}
$\left \| P^{\perp}_{\mathcal{R}(A)}b \right \|/\hat \sigma$ follows the chi distribution with $m-d$ degrees of freedom. Then, the proof is completed by averaging the chi and chi-squared distributions.
\end{proof}

\begin{lem}\label{Lemma2}
Suppose that every $m$ columns in $\Phi$ exhibit full rank, $y$ is uniformly sampled from $\mathcal{R}(\Phi_{\Omega})$, and the sparsity $k$ is smaller than $m$, i.e., $|\Omega| < m$. Then, in the noiseless case, the following statements (a) and (b) hold for any set $\Delta \subseteq \{1:n\}$ such that $|\Delta| < m$.
\begin{enumerate}[(a)]
\item $\Delta \nsupseteq \Omega$  holds  if  $\left \| P^{\perp}_{\mathcal{R}(\Phi_{\Delta})}y \right\| > 0$.
\item $\Delta \supseteq \Omega$ holds almost surely  if  $\left \| P^{\perp}_{\mathcal{R}(\Phi_{\Delta})}y \right\| = 0$. 
\end{enumerate}
\end{lem}

\begin{proof}
The following subspace $E$ is defined as 
\begin{align}
E:=\cup_{\Gamma \subseteq \{1:n\} \textup{ s.t. } |\Gamma \cap \Omega|<|\Omega| \textup{ and }  |\Gamma|<m } \,\mathcal{R}(\Phi_{\Gamma}).
\end{align}
For any set $\Delta \subseteq \{1:n\}$ such that $|\Delta|<m$, $y \notin \mathcal{R}(\Phi_{\Delta})$ holds if $\left \| P^{\perp}_{\mathcal{R}(\Phi_{\Delta})}y \right\| > 0$. If $\Delta \supseteq \Omega$, $\mathcal{R}(\Phi_{\Delta})$ includes $ \mathcal{R}(\Phi_{\Omega})$ so that $y \in \mathcal{R}(\Phi_{\Delta})$. Thus, $y \notin \mathcal{R}(\Phi_{\Delta})$ implies that $\Delta \nsupseteq \Omega$ so that  the statement (a) is satisfied.

Given that from the assumption for $\Phi$, the rank of $\mathcal{R}(\Phi_{\Omega})$ is strictly larger than that of $\mathcal{R}(\Phi_{\Gamma}) \cap \mathcal{R}(\Phi_{\Omega})$ for any index set $\Gamma \subseteq \{1:n\}$ such that  $|\Gamma \cap \Omega|<|\Omega |$ and $ |\Gamma|<m$, the event region satisfying $y \in \mathcal{R}(\Phi_{\Omega}) \cap E$ has Lebesgue measure zero on the range space $\mathcal{R}(\Phi_{\Omega})$. Thus, the condition $y \in \mathcal{R}(\Phi_{\Omega}) \setminus E$ holds almost surely. 
Given that the condition $\left \| P^{\perp}_{\mathcal{R}(\Phi_{\Delta})}y \right\| =0$ implies that $y \in \mathcal{R}(\Phi_{\Delta})$, $y \in \mathcal{R}(\Phi_{\Delta}) \cap (\mathcal{R}(\Phi_{\Omega}) \setminus E)$ holds almost surely. And this condition $y \in \mathcal{R}(\Phi_{\Delta}) \cap (\mathcal{R}(\Phi_{\Omega}) \setminus E)$ implies that ${\Delta} \supseteq \Omega$  from the definition of $E$.  Thus, the statement (b) is satisfied.
\end{proof}

\balance
\bibliographystyle{IEEEtran}
\bibliography{IEEEabrv,bibdata4}

\begin{thebibliography}{10}
\providecommand{\url}[1]{#1}
\csname url@samestyle\endcsname
\providecommand{\newblock}{\relax}
\providecommand{\bibinfo}[2]{#2}
\providecommand{\BIBentrySTDinterwordspacing}{\spaceskip=0pt\relax}
\providecommand{\BIBentryALTinterwordstretchfactor}{4}
\providecommand{\BIBentryALTinterwordspacing}{\spaceskip=\fontdimen2\font plus
\BIBentryALTinterwordstretchfactor\fontdimen3\font minus
  \fontdimen4\font\relax}
\providecommand{\BIBforeignlanguage}[2]{{%
\expandafter\ifx\csname l@#1\endcsname\relax
\typeout{** WARNING: IEEEtran.bst: No hyphenation pattern has been}%
\typeout{** loaded for the language `#1'. Using the pattern for}%
\typeout{** the default language instead.}%
\else
\language=\csname l@#1\endcsname
\fi
#2}}
\providecommand{\BIBdecl}{\relax}
\BIBdecl

\bibitem{candes2006compressive}
E.~J. Cand{\`e}s \emph{et~al.}, ``Compressive sampling,'' in \emph{Proceedings
  of the International Congress of Mathematicians}, vol.~3.\hskip 1em plus
  0.5em minus 0.4em\relax Madrid, Spain, 2006, pp. 1433--1452.

\bibitem{donoho2006compressed}
D.~L. Donoho, ``{C}ompressed sensing,'' \emph{IEEE Transactions on Information
  Theory}, vol.~52, no.~4, pp. 1289--1306, 2006.

\bibitem{metzler2016denoising}
C.~A. Metzler, A.~Maleki, and R.~G. Baraniuk, ``From denoising to compressed
  sensing,'' \emph{IEEE Transactions on Information Theory}, vol.~62, no.~9,
  pp. 5117--5144, 2016.

\bibitem{yang2010image}
J.~Yang, J.~Wright, T.~S. Huang, and Y.~Ma, ``Image super-resolution via sparse
  representation,'' \emph{IEEE Transactions on Image Processing}, vol.~19,
  no.~11, pp. 2861--2873, 2010.

\bibitem{heckel2016super}
R.~Heckel, V.~I. Morgenshtern, and M.~Soltanolkotabi, ``Super-resolution
  radar,'' \emph{Journal of the Institute of Mathematics and its Applications
  (IMA)}, vol.~5, no.~1, pp. 22--75, 2016.

\bibitem{lee2016map}
N.~Lee, ``{MAP} support detection for greedy sparse signal recovery algorithms
  in compressive sensing,'' \emph{IEEE Transactions on Signal Processing},
  vol.~64, no.~19, pp. 4987--4999, 2016.

\bibitem{wang2012generalized}
J.~Wang, S.~Kwon, and B.~Shim, ``{G}eneralized orthogonal matching pursuit,''
  \emph{IEEE Transactions on Signal Processing}, vol.~60, no.~12, pp.
  6202--6216, 2012.

\bibitem{needell2009cosamp}
D.~Needell and J.~A. Tropp, ``{C}o{S}a{MP}: {I}terative signal recovery from
  incomplete and inaccurate samples,'' \emph{Applied and Computational Harmonic
  Analysis}, vol.~26, no.~3, pp. 301--321, 2009.

\bibitem{dai2009subspace}
W.~Dai and O.~Milenkovic, ``Subspace pursuit for compressive sensing signal
  reconstruction,'' \emph{IEEE transactions on Information Theory}, vol.~55,
  no.~5, pp. 2230--2249, 2009.

\bibitem{davies2012rank}
M.~E. Davies and Y.~C. Eldar, ``{R}ank awareness in joint sparse recovery,''
  \emph{IEEE Transactions on Information Theory}, vol.~58, no.~2, pp.
  1135--1146, 2012.

\bibitem{kwon2014multipath}
S.~Kwon, J.~Wang, and B.~Shim, ``Multipath matching pursuit,'' \emph{IEEE
  Transactions on Information Theory}, vol.~60, no.~5, pp. 2986--3001, 2014.

\bibitem{rossi2012spatial}
M.~Rossi, A.~M. Haimovich, and Y.~C. Eldar, ``Spatial compressive sensing in
  {MIMO} radar with random arrays,'' in \emph{46th Annual Conference on
  Information Sciences and Systems (CISS)}, 2012, pp. 1--6.

\bibitem{wipf2004sparse}
D.~P. Wipf and B.~D. Rao, ``Sparse {B}ayesian learning for basis selection,''
  \emph{IEEE Transactions on Signal Processing}, vol.~52, no.~8, pp.
  2153--2164, 2004.

\bibitem{foucart2013mathematical}
S.~Foucart and H.~Rauhut, \emph{{A} mathematical introduction to compressive
  sensing}.\hskip 1em plus 0.5em minus 0.4em\relax Springer, 2013.

\bibitem{wang2016joint}
B.~Wang, L.~Dai, T.~Mir, and Z.~Wang, ``Joint user activity and data detection
  based on structured compressive sensing for {NOMA},'' \emph{IEEE
  Communications Letters}, vol.~20, no.~7, pp. 1473--1476, 2016.

\bibitem{cirik2017multi}
A.~C. Cirik, N.~M. Balasubramanya, and L.~Lampe, ``Multi-user detection using
  {ADMM}-based compressive sensing for uplink grant-free {NOMA},'' \emph{IEEE
  Wireless Communications Letters}, 2017.

\bibitem{tse2005fundamentals}
D.~Tse and P.~Viswanath, \emph{Fundamentals of wireless communication}.\hskip
  1em plus 0.5em minus 0.4em\relax Cambridge university press, 2005.

\end{thebibliography}

\ifCLASSOPTIONcaptionsoff
  \newpage
\fi

\end{document}